\documentclass[12pt,twoside,draft]{cmpart}
\usepackage[final]{graphicx}
\usepackage{amsmath}
 \usepackage[russian,ukrainian,english]{babel}
\author{Boyko O$^1$, Kaliuzhnyi-Verbovetskyi D$^2$ , Pivovarchik V$^3$}
\title{Non-existence of co-spectral simple connected graphs with small number of edges }

\shorttitle{Non-existence of co-spectral simple connected graphs... }

\institute{$^{1,2}$South Ukrainian National Pedagogical Univeristy, Staroportofrankovskaya str. 26, Odesa \\
$^{3}$South Ukrainian National Pedagogical Univeristy, Staroportofrankovskaya str. 26, Odesa and Univeristy of Vaasa, Vaasa, Finland}

\email{$^{1}$  boykohelga@gmail.com, $^{2}$ dmitry2k@yahoo.com, $^3$vpivovarchik2@gmail.com }

\enabstract{Author1 A.A., Author2 B.B.}%
{Title of the article in English}%
{ We consider the Sturm-Liouville spectral problems on simple connected equilateral  graphs.  
We prove that if the number of edges does not exceed 7 then the asymptotics of eigenvalues of the  Dirichlet problem uniquely determine the shape of the graph. }
{Graph,  edge, vertex, eigenvalue, potential, non-isomorphic,   adjacency matrix, characteristic function}

\uaabstract{�. �����, �. ��������-������������, �. ����������}%
{�� ��������� ��-������������ ������� ��'����� ������ � ����� ��������� �����}%
{������������� ������ ������-˳������ �� ������� ��'����� ����������� ������. ��������, ��, ���� ��������� ����� �� ��������� 7, �� ����������� ������� ������� ������ ĳ����� ���������� �������� ����� �����.
}{����, �����, �������, ������ ��������, ���������, �����������, ������� ����������, ��������������� �������}

\thanks{ The authors are grateful to the Ministry of Education and Science of Ukraine for the support in completing the work 'Inverse problems of finding the shape of a graph by spectral data' 
State registration number 0124U000818.

The present research was supported by the Academy of Finland (project no. 358155). The third author is grateful to the University of Vaasa for hospitality.

 The  first and the third authors express their gratitude to NSF US for support IMPRESS-U: Spectral and geometric methods for damped wave equations with applications to fiber lasers.}

\UDC{?????}
\received{04,,???.02.2021}  % article received date
\begin{document}

\maketitle

%%% ----------------------------------------------------------------------

\section*{Introduction}
In quantum graph theory, the problem of recovering the shape of a graph was stated  in \cite{vB} and \cite{GS}. It was shown in \cite{GS} that if the lengths of the edges are non-commensurate then the spectrum of the spectral Sturm-Liouville problem on a graph with the standard conditions at its vertices (the Neumann conditions at the pendant vertices and the  continuity conditions + Kirchhoff's conditions at the interior vertices)   uniquely determines the shape of this graph.  In \cite{vB}, it was shown that in the case of commensurate lengths of the edges there exist co-spectral quantum graphs. 
In \cite{KN}, \cite{BKS} it was shown that the spectrum of the Neumann problem  with zero potential  on $P_2$ uniquely determines the shape of the graph.  
In \cite{CP}, it was shown that if the graph is simple connected equilateral with the number of vertices less than or equal 5 and the potentials on the edges are real $L_2$ functions, then the spectrum of the Sturm-Liouville problem with  the standard conditions at the vertices uniquely determines the shape of the graph. For trees, the minimal number of vertices in a co-spectral pair is 9 (see \cite{Pist},  \cite{CP1}, \cite{BMP}). In the case of the standard conditions at all the vertices, the asymptotics of the spectrum shows whether the graph is a tree. If the number of vertices doesn't exceed 8 then to find the shape of a tree we need just to find in \cite{CP} the characteristic polynomial corresponding to the given spectrum. 

For the case of the Dirichlet conditions at the pendant vertices, it was shown in \cite{BMP} that there are no cospectral trees with the number of vertices $\leq 8$. However, in case of the Dirichlet conditions at the pendant vertices, the asymptotics of the spectrum do not show whether the graph is a tree. The  possibility  of coincidence between the spectrum of a tree and the spectrum of a nontree graph (or  between the spectra of two graphs with different cyclomatic numbers) is not excluded  a'priori in  the case of the Dirichlet conditions at the pendant vertices. Such possibility is excluded in the case of the Neumann conditions at the pendant vertices.

It should be mentioned that an attempt to use two spectra to find the shape of a tree has been done in \cite{Piv} where  expantion to a branched continued fractions of certain polynomials related to the Neumann and the Dirichlet problems was used. For branched continued fractions see \cite{DLD}.  

 In  the present paper,  we prove that if the number of edges $g\leq 7$ then
there are no co-spectral simple connected  equilateral graphs and we conclude that   the  asymptotics of the spectrum uniquely determine the shape of the graph. Examples of co-spectral graphs with  the  Dirichlet conditions at the pendant vertices with $g=8$ are given in \cite{Pist} (see Fig 1).

In Section 2, we describe  the Sturm-Liouville problem on a simple connected equilateral graph  with  the Dirichlet  conditions at each pendant vertex and the standard conditions at all the interior vertices. 

We expose  the known theorem  relating the characteristic function (the function whose set of zeros coincides with the spectrum) of the  above described Sturm-Liouville problem  with the  determinant of the normalized Laplacian of the corresponding combinatorial graph.

In Section 3,  we  show all the simple connected equilateral graphs on 7 or less edges and calculate the corresponding characteristic functions. We compare them and find that there are no co-spectral pairs.

\section{Statement of the problem}

Let $G$ be an equilateral simple connected graph with $p$ vertices and $g$ edges each of the length $l$.

We  direct  the edges incident with the pendant vertices  away from these vertices. Orientation of the rest of the vertices is arbitrary.  
Let us describe the spectral  problem on $G$. We consider  the Sturm-Liouville equations on the edges

\begin{equation}
\label{2.1}
-y_j^{\prime\prime}+q_j(x)y_j=\lambda y_j, \ \  j=1,2,..., g 
\end{equation} 
where $q_j\in L_2(0,l)$ are real. 

At the beginning  of  each edge $e_j$ incident with a pendant vertex,  
we impose the Dirichlet condition 
\begin{equation}
\label{2.2}
y_j(0)=0.
\end{equation}

At  each interior vertex, we impose the standard conditions, i.e. the continuity conditions

\begin{equation}
\label{2.3}
y_j(l)=y_k(0)
\end{equation}
for the incoming into $v_i$ edges $e_j$ and for all $e_k$ outgoing from $v_i$, and the Kirchhoff's conditions
\begin{equation}
\label{2.4}
\mathop{\sum}\limits_jy'_j(l)=\mathop{\sum}\limits_k y_k'(0),
\end{equation}
where the sum in the right-hand side is taken over all edges $e_k$ outgoing from $v_i$ and the sum in the left-hand side is taken over all edges $e_j$ incoming to $v_i$.

We call the above conditions (the continuity + Kirchhoff's or Neumann's) standard.

In the sequel, if the potentials are the same on all the  edges we omit the index in $q_j$ and $y_j$. 

In order to find the characteristic function of our Sturm--Liouville problems,  we look for coefficients $A_j, B_j$ such that  the solution of \eqref{2.1} can be expressed in the form 
\[
y_j=A_{j}s_{j}(\lambda,x)+B_{j}c_{j}(\lambda,x), \  x\in (0,l).
\]
Substituting this into the continuity conditions; as well as into Kirchhoff's condition at each interior vertex and into the Dirichlet conditions at all pendant vertices, 
we obtain a system of $2g$ linear algebraic equations with unknowns $A_j, B_j$.
Denote the $2g\times 2g$ matrix of this system by  $\|\Phi_D(\lambda)\|$,  we call it the \textit{characteristic matrix} of our problem. Observe that it involves the values $s_{j}(\lambda,l)$, $s'_{j}(\lambda,l)$, $c_{j}(\lambda,l)$, $c'_{j}(\lambda, l)$. 
Then the equation
\[
\det\|\Phi_D(\lambda)\|=0
\]
completely determines the spectrum of problem \eqref{2.1}--(\ref{2.4}).

Let $A$ be the adjacency matrix of $G$, and
\[D={\rm diag}(d(v_0), d(v_1),...,d(v_{p})),\]  the degree matrix. Here 
$d(v_i)$ is the degree of the vertex $v_i$.  Denote by $-z\hat{D}+\hat{A}$ the submatrix of $-zD+A$ obtained by deleting the rows and the columns corresponding to those pendant vertices (where the Dirichlet conditions are imposed).

The following theorem was proved in \cite{MP2} (Theorem 6.4.2).

\begin{theorem}  
 Let $G$ be a simple connected graph  with $p\geq 2$.  Assume that all edges have the same length $l$ and the same real potential $q(x)\in L_2(0,l)$ symmetric with respect to the midpoint of an edge ($q(x)=q(l-x)$).  Then the spectrum of problem (\ref{2.1})--(\ref{2.4})   coincides with the set of zeros of the characteristic   function 
\begin{equation}
\label{2.5}
\phi(\lambda)=s^{g-p+r}(\lambda,l)\psi(c(\lambda,l)),
\end{equation}
where $r$ is the number of pendant vertices, $s(\lambda,x)$ is the solution of (\ref{2.1}) which satisfies the conditions $s(\lambda,0)=s'(\lambda,0)-1=0$ and $c(\lambda,x)$ is the solution of (\ref{2.1}) which
satisfies the conditions $c(\lambda,0)-1=c'(\lambda,0)=0$ and
\[\psi(z)=det(-z\hat{D}+\hat{A}).\] 
\end{theorem}

It is clear that,   in case of identically zero potential, 
\begin{equation}
\label{2.6} 
\phi_0(\lambda)=\left(\frac{\sin\sqrt{\lambda}l}{\sqrt{\lambda}}\right)^{g-p+r}\psi(\cos\sqrt{\lambda} l). 
\end{equation} 

%\end{document} 
\begin{corollary} Let $G$ be a simple connected  graph with $p\geq 2$. Assume that the edges have the same length $l$ and the potentials on the edges $q_j(x)\in L_2(0,l)$ are real. 
Then the characteristic function of problem (\ref{2.1})--(\ref{2.4}) satisfies
\begin{equation}
\label{2.7}
\phi(\lambda)\mathop{=}\limits^{\lambda\rightarrow + \infty}\phi_0(\lambda)+O(1).
\end{equation}
\end{corollary}
\begin{proof} We use the following asymptotics \cite{Ma}:
\begin{equation}
\label{2.8}
s_j(\lambda,l)\mathop{=}\limits^{\lambda\rightarrow+\infty}\frac{\sin\sqrt{\lambda}l}{\sqrt{\lambda}}+O\left(\frac{1}{\lambda}\right), \ \ \ 
c_j(\lambda,l)\mathop{=}\limits^{\lambda\rightarrow+\infty}\cos\sqrt{\lambda}l+O\left(\frac{1}{\sqrt{\lambda}}\right). 
\end{equation}
\[
s''_j(\lambda,l)\mathop{=}\limits^{\lambda\rightarrow+\infty}\cos\sqrt{\lambda}l+O\left(\frac{1}{\sqrt{\lambda}}\right), \ \ \ 
c'_j(\lambda,l)\mathop{=}\limits^{\lambda\rightarrow+\infty}-\sqrt{\lambda}\sin\sqrt{\lambda}l+O\left(1\right).
\]
Suppose first that all the potentials on the edges are the same and symmetric with respect to  the midpoint of an edge. Then using (\ref{2.5}) and (\ref{2.8}) we obtain (\ref{2.7}). 
\end{proof}
Now let the potentials be diffirent and not symmetric but real $L_2(0,l)$-functions. Then we apply Theorem 5.4  from  \cite{CaP2} and obtain (\ref{2.7}). This means that if two graphs are cospectral then they must have not only the same $\phi(z)$ but also the same $\phi_0(z)$. Thus, we need to investigate $\phi_0(z)$.

\section{Cospectrality}
The spectrum of problem (\ref{2.1})-- (\ref{2.4}) consists of normal (isolated Fredholm) eigenvalues of finite multiplicity. The corresponding operator is selfadjoint,  therefore these eigenvalues are real. For the main term of the asymptotics we have

%\end{document}

\begin{equation}
\label{55}
\mathop{\lim}\limits_{k \to \infty}\frac{\lambda_k}{k^2}=\frac{\pi^2}{g^2l^2}
\end{equation}
(see \cite{Meh}, \cite{Bel85}, \cite{Ni} or \cite{vB} Corollary 1).
%\end{document}
In this paper,  by co-spectral we mean  simple connected equilateral (with the same length of the edges) graphs with the same spectrum of problem (\ref{2.1})--(\ref{2.4}).

According to (\ref{55}),   co-spectral graphs must have the same number of edges,  $g$. 

Equations  (\ref{2.6}) and (\ref{2.7}) imply  that co-spectral graphs must have the same value of $g-p+r$. Let us explain it. The following theorem was proved in \cite{BMP}.

\begin{theorem}   
 Let $T$ be an equilateral tree. The eigenvalues of problem (\ref{2.1})--(\ref{2.4})  can be presented as the union of subsequences $\{\lambda_k\}_{k=1}^{\infty}=\mathop{\bigcup }\limits_{i=1}^{2p-r-1}\{\lambda_k^{(i)}\}_{k=1}^{\infty}$
with the following asymptotics: 
\begin{equation}
\label{3.1}
\sqrt{\lambda_k^{(i)}}\mathop{=}\limits_{k\to\infty}\frac{2\pi (k-1)}{l}\pm\frac{1}{l}\arccos \alpha_{i}+O\left(\frac{1}{k}\right) \ \ {\rm for}  \ \  i=1,2,..., p-p_{pen}, \ \ k=1,2,...
\end{equation}
\begin{equation}
\label{3.3}
\sqrt{\lambda_k^{(i)}}\mathop{=}\limits_{k\to\infty}\frac{\pi k}{l}+O\left(\frac{1}{k}\right) \ \  {\rm for} \ \  i=p-p_{pen}+1, ..., p-1, \ \ k=1,2,...
\end{equation}
Here $\alpha_1, \alpha_2, ..., \alpha_{p-p_{pen}}$ are the zeros of  the polynomial $\psi{}(z)$, and $p_{pen}=r$ is the number of pendant vertices.
\end{theorem}
%\end{document}
Here the subsequences (\ref{3.3}) correspond to the factor $\left(\frac{\sin\sqrt{\lambda} l}{\sqrt{\lambda}}\right)^{q-p+r}$ in  (\ref{2.6})       while the subsequences (\ref{3.1}) correspond to the factor $\psi(\cos\sqrt{\lambda}l)$.
Of course, the polynomial $\psi(z)$ may contain a factor $(1-z^2)$ and, consequently, $\psi(\cos \sqrt{\lambda} l)$  may contain $(1-\cos^2\sqrt{\lambda})=\sin^2\sqrt{\lambda} l$. However, this factor gives a subsequence
\[
\sqrt{{\lambda}_k^{(i)}}\mathop{=}\limits_{k\to\infty}\frac{\pi k}{l}+O\left(\frac{1}{k}\right) \ \  {\rm for} \ \ k=0,1,...
\]
which starts with $k=0$ and therefore differs from (\ref{3.3}). Similar arguments lead to the assertion that to check cospectrality it is sufficient to compare only graphs with the same  $g$ and the same $\Delta\equiv g-p+r$.  
%\end{document}
\begin{theorem} There are no co-spectral graphs among simple connected equilateral graphs of seven or less edges in case of the Dirichlet conditions at the pendant vertices and standard conditions at the interior vertices. 
\end{theorem}
\begin{proof} All simple connected graphs of 1, 2, 3 and 4 edges are presented at Fig.  1. We denote the graphs by $G^i_{g,g-p+r}$ where the upper index enumerate the graphs with the same $g$ and same $g-p+r$ given as the lower indices. 
\begin{figure}[h]
  \begin{center}
   \includegraphics[scale=0.5 ] {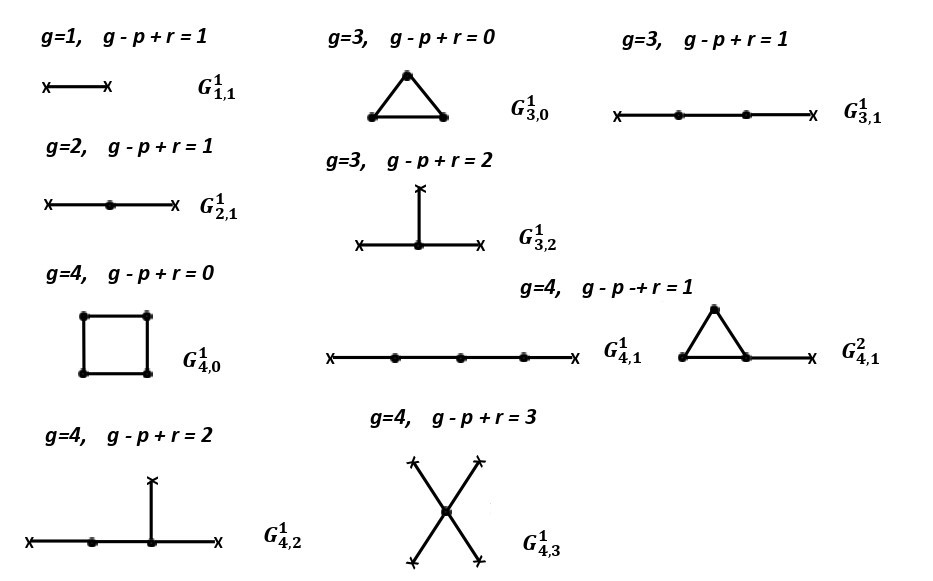}
  \end{center}
\caption{  Simple connected graphs of 1,2,3 and 4 edges}
\end{figure}
%\end{proof}
%\end{document}
Among these graphs there are only two  $G_{4,1}^1$ and $G_{4,1}^2$  with the same $g$ and $g-p+r$. However, the corresponding characteristic polynomials are 
\[
\psi_{4,1}^1(z)=-8z^3+4z,  \ \ \ \psi_{4,1}^2(z)=-12z^3+7z+2
\]
The sets of zeros of these polynomials are different. Thus there are no co-spectral graphs among the graphs of Fig. 1. 

Now let us consider the graphs of 5 edges. All simple connected graphs of 5 edges are presented at Fig 2.

\begin{figure}[h]
  \begin{center}
  \includegraphics[scale=0.5 ] {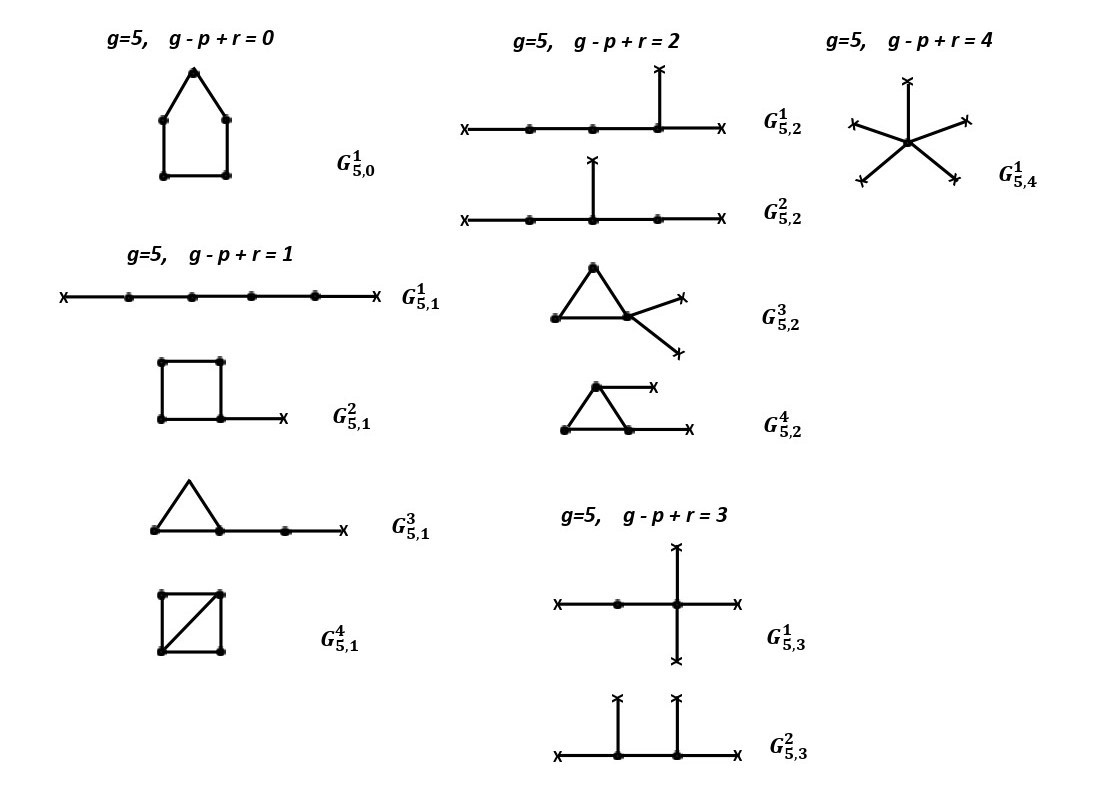}
  \end{center}
\caption{Simple connected graphs with $g=5$}
\end{figure}

%\end{proof}

Among these graphs there are four graphs  $G_{5,1}^1$, $G_{5,1}^2$, $G_{5,1}^3$, and $G_{5,1}^4$  with the same $g=5$ and $g-p+r=1$. However, the corresponding characteristic polynomials are 
\[\psi_{5,1}^1(z)=16z^4-12z^2+1,  \ \ \ \psi_{5,1}^2(z)=24z^4-20z^2,
\]
\[\psi_{5,1}^3(z)=24z^4-18z^2-4z+1,  \ \ \ \psi_{5,1}^4(z)=24z^4-24z^2-8z,
\]
It is clear that these polynomials have different sets of zeros.

There are  four graphs  $G_{5,2}^1$, $G_{5,2}^2$, $G_{5,2}^3$,  and $G_{5,2}^4$  with the same $g=5$ and $g-p+r=2$. Their polynomials are
\[\psi_{5,2}^1(z)=-12z^3+5z,  \ \ \ \psi_{5,2}^2(z)=-12z^3+4z,
\]
\[\psi_{5,2}^3(z)=-16z^3+8z+2,  \ \ \ \psi_{5,2}^4(z)=-18z^3+8z+2,
\]
The sets of zeros of these polynomials are different.

There are two graphs $G_{5,3}^1$ and $G_{5,3}^2$ shown in Fig. 2 with $g=5$ and $g-p+r=3$. Their characteristic polynomials 
\[\psi_{5,3}^1(z)=8z^2-1,  \ \ \ \psi_{5,3}^2(z)=9z^2-1
\]
have different sets of zeros.

There are 29 simple connected graphs with 6 edges. They are shown in Fig. 3 and Fig. 4. 
 
\begin{figure}[h]
\begin{center}
\includegraphics[scale=0.5 ] {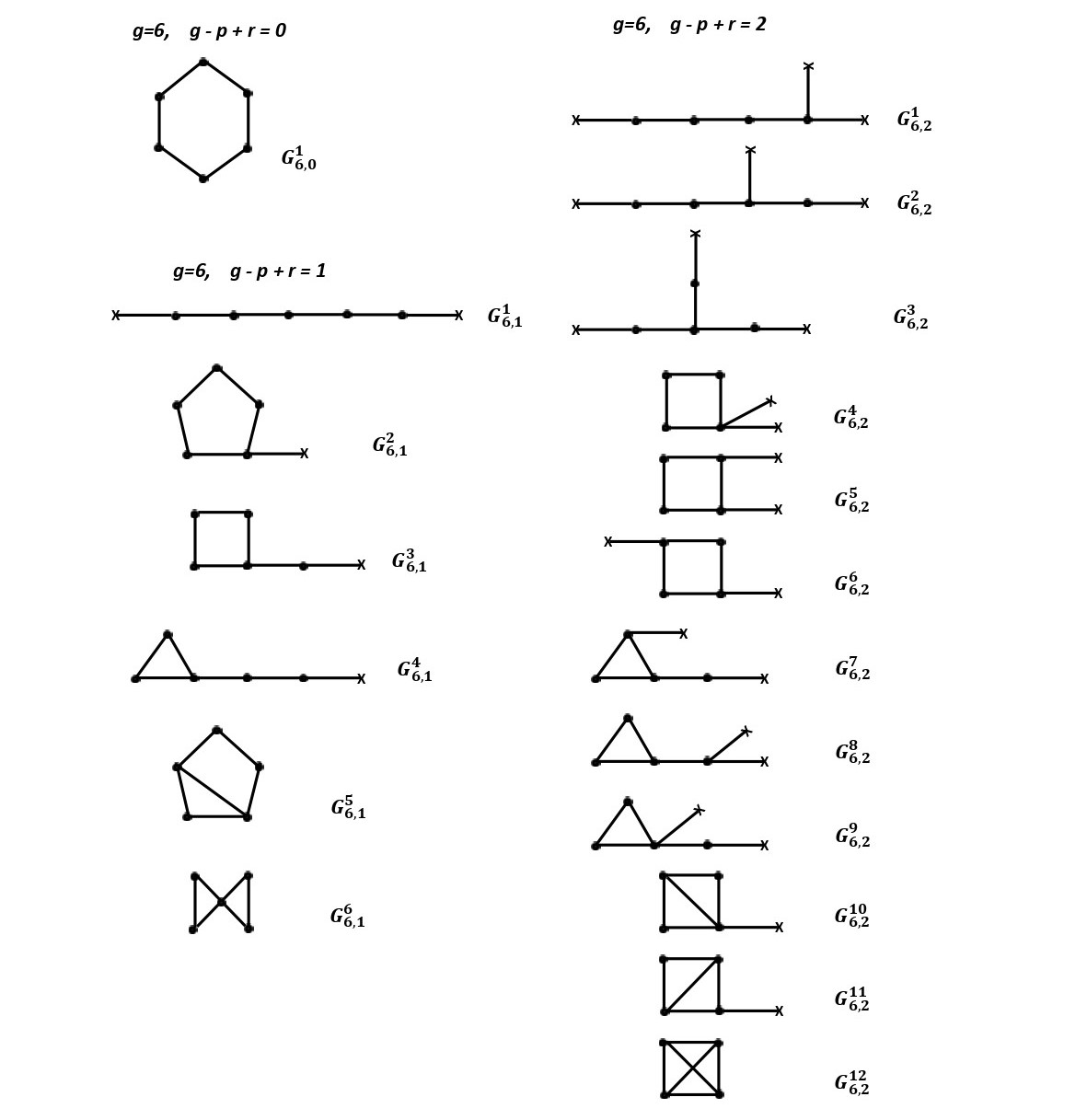}
 \end{center}
\caption{Simple connected graphs of 6 edges}
\end{figure}

%\newpage

\begin{figure}[h]
\begin{center}
\includegraphics[scale=0.5 ] {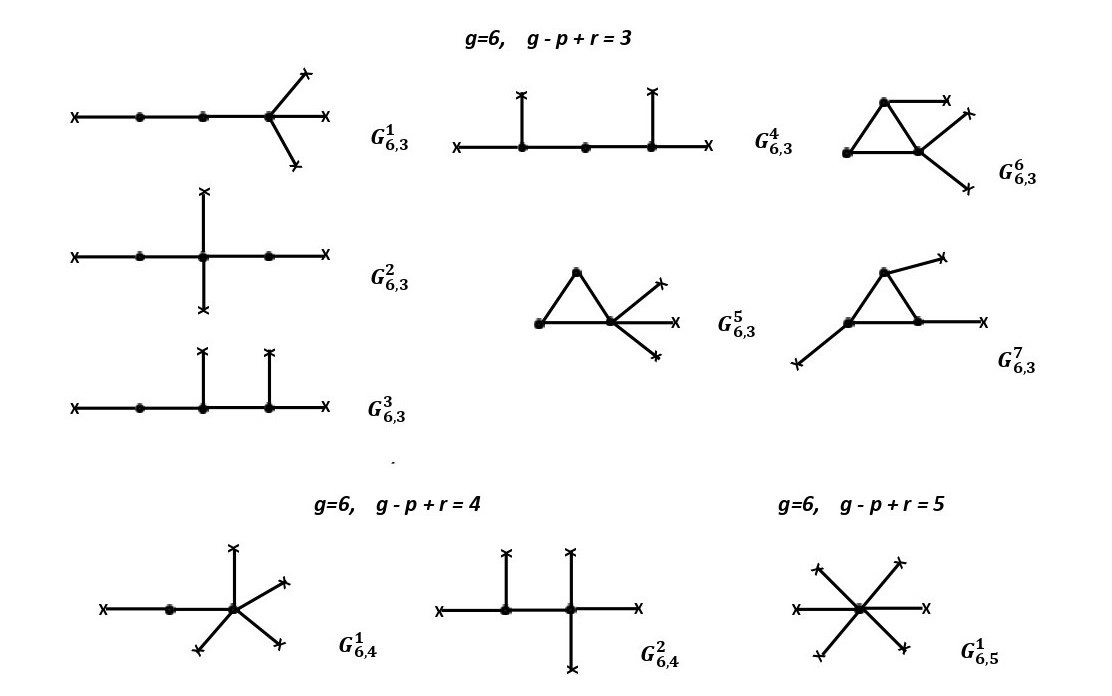}
\end{center}
\caption{Simple connected graphs of 6 edges}
\end{figure}

%\newpage

We see that there are six simple connected graphs $G_{6,1}^1$, $G_{6,1}^2$, $G_{6,1}^3$, $G_{6,1}^4$. $G_{6,1}^5$, $G_{6,1}^6$ 
with $g=6$ and $g-p+r=1$.  Their characteristic polynomials are

\[
\psi_{6,1}^1(z)=-32z^5+32z^3-6z, \ \ \ \psi_{6,1}^2(z)=-48z^5+52z^2-11z+2,
\]
\[
\psi_{6,1}^3(z)=-48z^5+48z^3-4z, \ \ \ \psi_{6,1}^4(z)=-48z^5+48z^3+8z^2-9z-2,
\]
\[
\psi_{6,1}^5(z)=-48z^5+60z^3+8z^2-9z, \ \ \ 
\psi_{6,1}^6(z)=-64z^5+64z^3+16z^2-12z-4.
\]
We see that there are no polynomials with the same set of zeros among these.
%\end{proof}
%\end{document}

There  are 12 graphs $G_{6,2}^1$, $G_{6,2}^2$, $G_{6,2}^3$, $G_{6,2}^4$. $G_{6,2}^5$, $G_{6,2}^6$, $G_{6,2}^7$, $G_{6,2}^8$, $G_{6,2}^9$, $G_{6,2}^{10}$. $G_{6,2}^{11}$, $G_{6,2}^{12}$ with $g=6$ and $g-p+r=2$. The corresponding polynomials are 
\[
\psi_{6,2}^1(z)=24z^4-16z^2+1, \ \ \ 
\psi_{6,2}^2(z)=24z^5-14z^2+1.
\]
\[
\psi_{6,2}^3(z)=24z^4-12z^2, \ \ \ 
\psi_{6,2}^4(z)=32z^4-24z^2,
\]
\[
\psi_{6,2}^5(z)=36z^4-25z^2, \ \ \
\psi_{6,2}^6(z)=36z^4-24z^2,
\]
\[
\psi_{6,2}^7(z)=36z^4-22z^2-4z+1, \ \ \
\psi_{6,2}^8(z)=36z^5-25z^2-6z+1,
\]
\[
\psi_{6,2}^9(z)=32z^4-20z^2-4z+1, \ \ \
\psi_{6,2}^{10}(z)=48z^--32z^2-8z.
\]
\[
\psi_{6,2}^{11}(z)=54z^4-36z^2-10z, \ \ \
\psi_{6,2}^{12}(z)=81z^4-54z^2-24z-3.
\]
We see that the sets of zeros of these polynomials are different.

There are seven graphs with $g=6$ and $g-p+r=3$ (see Fig. 4.). The corresponding polynomials are
\[
\psi_{6,3}^1(z)=-16z^3+6z, \ \ \
\psi_{6,3}^2(z)=-16z^3+4z,
\]
\[
\psi_{6,3}^3(z)=-18z^3+5z, \ \ \
\psi_{6,3}^4(z)=-18z^3+6z,
\]
\[
\psi_{6,3}^6(z)=-20z^3+9z+2, \ \ \
\psi_{6,3}^7(z)=-24z^3+9z+2, 
\]
\[
\psi_{6,3}^8(z)=-27z^3+9z+2.
\]
They have different sets of zeros.

Two graphs of $g=6$ and $g-p+r=4$ are shown in Fig. 4. Their polynomials $\psi_{6.4}^1$ and $\psi_{6,4}^2$ have different sets of zeros: 
\[
\psi_{6,4}^1(z)=10z^2-1, \ \ \  
\psi_{6,4}^2(z)=12z^2-1.
\]

Thus, we conclude that there are no co-spectral graphs of 6 or less edges.

Now we consider graphs with $g= 7$. We look for co-spectral  among simple connected graphs.   The graph $C_7$ ( the cycle of seven vertices) has $g=7$ and $g-p+r=0$. There are no other simple connected graphs with such parameters and consequently $C_7$ has no co-spectral  partner. 

Now let $q=7$ and $g-p+r=1$. 
There are 8 simple connected graphs with $g=7$ and $g-p+r=1$. These graphs are given at Fig. 5.

%\newpage
\begin{figure}[h]
\begin{center}
   \includegraphics[scale=0.5 ] {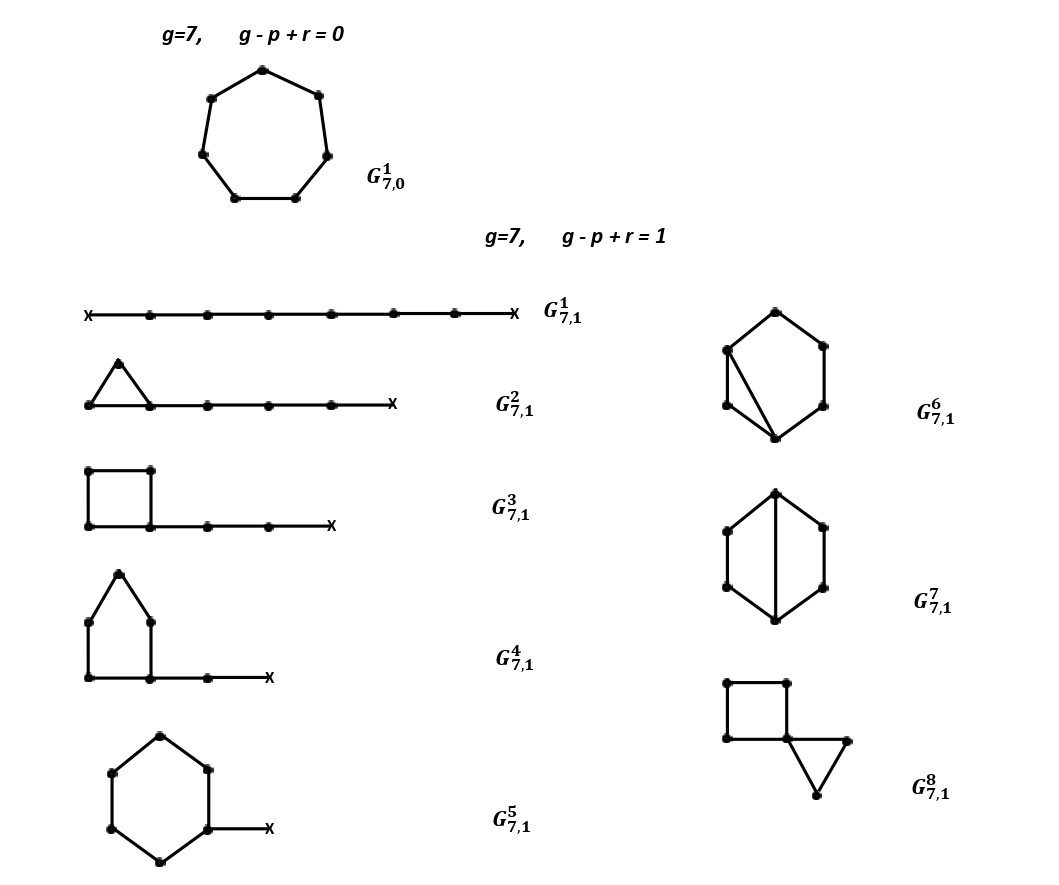}
 \end{center}
\caption{Simple connected graphs with $g=7$ and $g-p+r=1$.}
\end{figure}

%\end{proof}
%\end{document}
%\newpage

The corresponding characteristic polynomials are: 
\[
\psi_{7,1}^1(z)=64z^6-80z^4+24z^2-1,
\ \ \
\psi(z)_{7,1}^2=96z^6-120z^4-16z^3+36z^2+8z-1,
\]
\[
\psi(z)_{7,1}^3=96z^6-120z^4+28z^2,
\ \ \
\psi(z)_{7,1}^4=96z^6-112z^4+30z^2-2,
\]
\[
\psi(z)_{7,1}^5=96z^6-120z^4+34z^2-4z-1,
\ \ \
\psi(z)_{7,1}^6=144z^6-168z^4+48z^2-4,
\]
\[
\psi(z)_{7,1}^7=144z^6-168z^4+49z^2-4,
\ \ \
\psi(z)_{7,1}^8=128z^6-160z^4-16z^3+40z^2+8z.
\]
We see that there are no polynomials with the same set of zeros among them.

Let $q=7$ and $g-p+r=2$. 
There are 27 simple connected graphs  corresponding to $g=7$ and  $g-p+r=2$. These graphs  are shown at Fig. 6. The corresponding characteristic polynomials are

\begin{figure}[h]
\begin{center}
\includegraphics[scale=0.5] {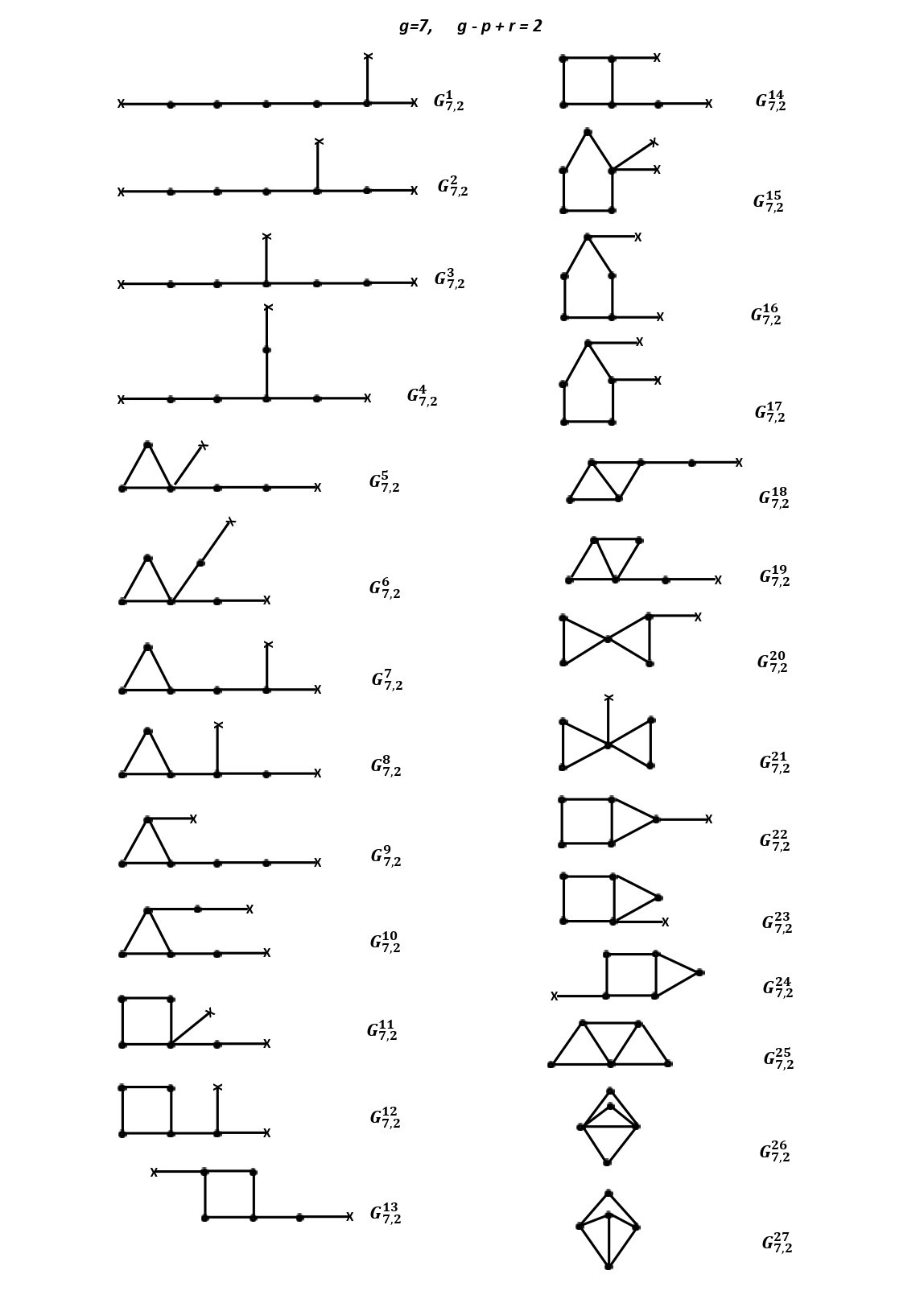}
\end{center}
\caption{Simple connected graphs with $g=7$ and $g-p+r=2$}
\end{figure}

\[
\psi_{7,2}^1(z)=-48z^5+44z^3-7z,
\ \ \
\psi_{7,2}^2(z)=-48z^5+40z^3-6z,
\]
\[
\psi_{7,2}^3(z)=-48z^5+40z^3-7z, \ \ \
\psi_{7,2}^4(z)=-48z^5+36z^3-4z.
\]
\[
\psi(z)_{7,2}^5=-64z^5+56z^3+8z^2-10z-2,
\ \ \ 
\psi(z)_{7,2}^6=-64z^5+48z^3+8z^2-4z,
\]
\[
\psi(z)_{7,2}^7=-72z^5+66z^3+12z^2-10z-2,
\ \ \
\psi(z)_{7,2}^8=-72z^5+62z^3+12z^2-9z-2,
\]
\[
\psi(z)_{7,2}^9=-72z^5+62z^3+8z^2-10z-2,
\ \ \
\psi(z)_{7,2}^{10}=-72z^5+56z^3+8z^2-6z,
\]
\[
\psi(z)_{7,2}^{11}=-64z^5+56z^3-4z,
\ \ \
\psi(z)_{7,2}^{12}=-72z^5+68z^3-4z,
\]
\[
\psi(z)_{7,2}^{13}=-72z^5+60z^3-4z,
\ \ \
\psi(z)_{7,2}^{14}=-72z^5+60z^3-5z,
\]
\[
\psi(z)_{7,2}^{15}=-64z^5+64z^3-12z+2,
\ \ \
\psi(z)_{7,2}^{16}=-72z^5+66z^3-12z+2,
\]
\[
\psi(z)_{7,2}^{17}=-72z^5+68z^3-12z+2,
\ \ \
\psi(z)_{7,2}^{18}=-108z^5+90z^3+2z^2-8z-2,
\]
\[
\psi(z)_{7,2}^{19}=-98z^5+76z^3+16z^2-4z,
\ \ \
\psi(z)_{7,2}^{20}=-96z^5+84z^3+20z^2-13z-4,
\]
\[
\psi(z)_{7,2}^{21}=-80z^5+72z^3+16z^2-13z-4,
\ \ \
\psi(z)_{7,2}^{22}=-108z^5+994z^3+8z^2-10z,
\]
\[
\psi(z)_{7,2}^{23}=-96z^5+88z^3+8z^2-11z,
\ \ \
\psi(z)_{7,2}^{24}=-108z^5+96z^3+12z^2-11z,
\]
\[
\psi(z)_{7,2}^{25}=-98z^5+176z^3+16z^2-4z,
\ \ \
\psi(z)_{7,2}^{26}=-128z^5+104z^3+24z^2,
\]
\[
\psi(z)_{7,2}^{27}=-16z2^5+144z^3+24z^2-6z.
\]
%\end{document}
We see that there are no plynomials with the same set of zeros among them.

%\newpage

\begin{figure}[h]
 \begin{center}
\includegraphics[scale=0.5] {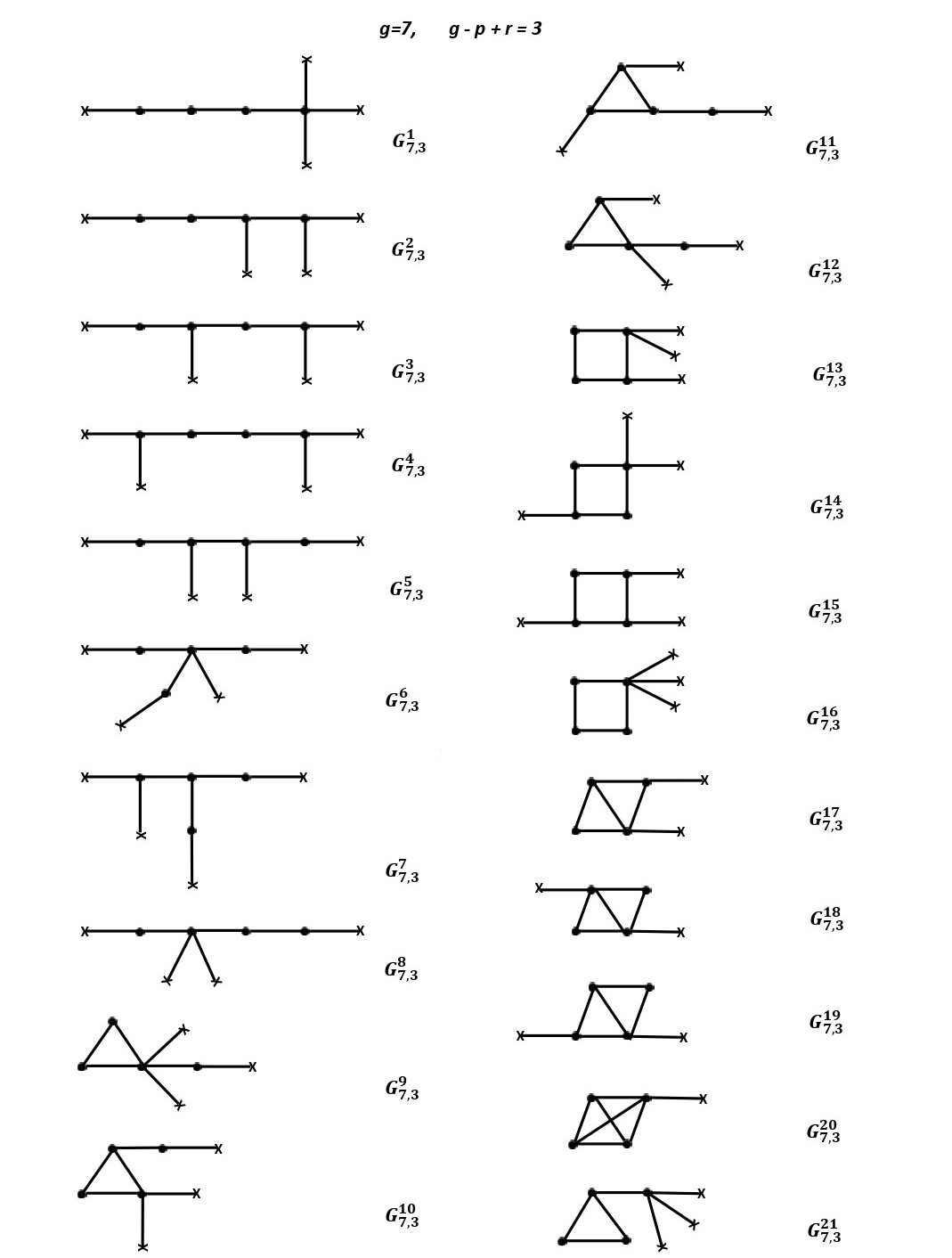}
  \end{center}
\caption{Simple connected graphs with $g=7$ and $g-p+r=3$}
\end{figure}

%\end{proof}
%\end{document}

We have 21 graphs corresponding to $g=7$ and $g-p+r=3$.   These graphs are  
shown at Fig. 7. 

The corresponding characteristic polynomials  are
\[
\psi(z)_{7,3}^{1}=32z^4-20z^2+1,
\ \ \
\phi(z)_{7,3}^{2}=36z^4-19z^2+1,
\]
\[
\psi(z)_{7,3}^{3}=36z^4-18z^2+1,
\ \ \
\psi(z)_{7,3}^{4}=36z^4-21z^2+1,
\]
\[
\psi(z)_{7,3}^{5}=36z^4-16z^2+1,
\ \ \
\psi(z)_{7,3}^{6}=36z^4-12z^2,
\]
\[
\psi(z)_{7,3}^{7}=36z^4-16z^2,
\ \ \
\psi(z)_{7,3}^{8}=36z^4-12z^2+1,
\]
\[
\psi(z)_{7,3}^{9}=
40z^4-22z^2-4z+1,
\ \ \
\psi(z)_{7,3}^{10}=
48z^4-26z^2-4z+1,
\]
\[
\psi(z)_{7,3}^{11}=
54z^4-27z^2-4z+1,
\ \ \
\psi(z)_{7,3}^{12}=
48z^4-24z^2-4z+1,
\]
\[
\psi(z)_{7,3}^{13}=
48z^4-30z^2,
\ \ \
\psi(z)_{7,3}^{14}=
48z^4-28z^2,
\]
\[
\psi(z)_{7,3}^{15}=
54z^4-30z^2,
\ \ \
\psi(z)_{7,3}^{16}=
40z^4-28z^2,
\]
\[
\psi(z)_{7,3}^{17}=
72z^4-41z^2-10z,
\ \ \
\psi(z)_{7,3}^{18}=
64z^4-36z^2-8z,
\]
\[
\psi(z)_{7,3}^{19}=
81z^4-76z^2-12z,
\ \ \
\psi(z)_{7,3}^{20}=
108z^4-63z^2-26z-3,
\]
\[
\psi(z)_{7,3}^{21}=
48z^--32z^2-8z+1.
\]
We see that there are no polynomials with the same set of zeros among them.
%\end{proof}
%\end{document}

We have 10 graphs corresponding to $g-p+r=4$, $g=7$. These graphs are  shown at Fig. 8. 
\begin{figure}[h]
  \begin{center}
   \includegraphics[scale=0.5] {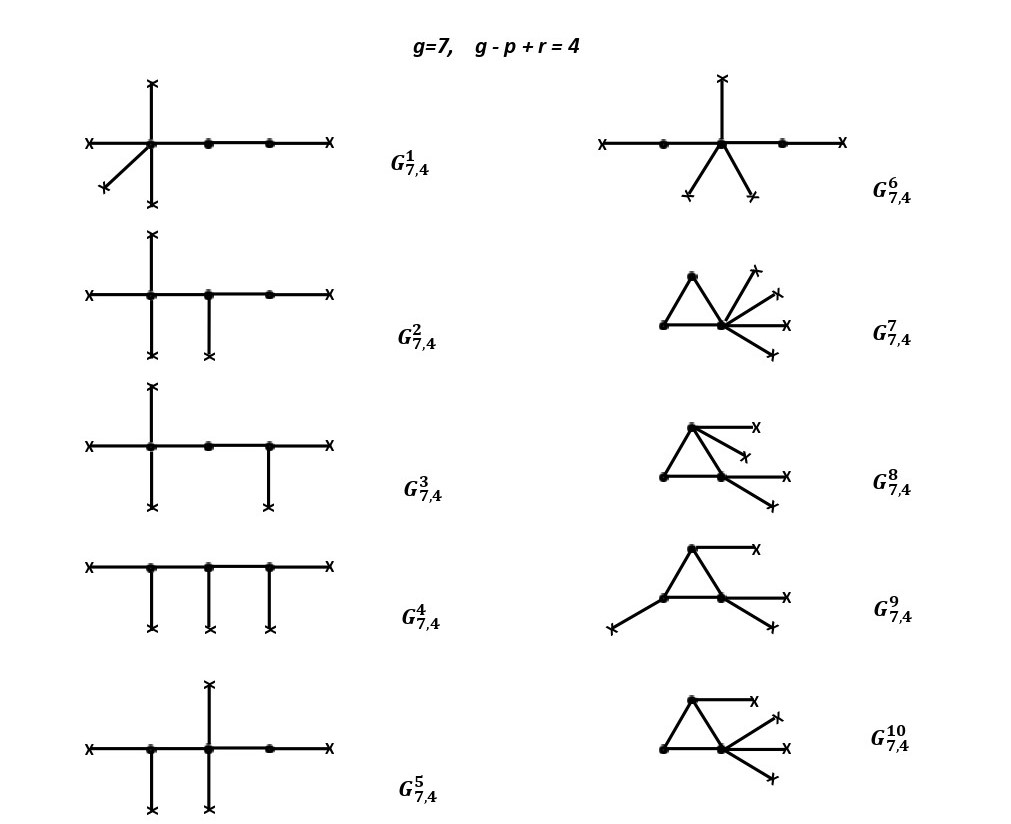}
  \end{center}
\caption{Simple connected graphs with $g=7$ and $g-p+r=4$}
\end{figure}
The corresponding characteristic polynomials  are
\[
\psi(z)_{7,4}^{1}=
-20z^3+7z,
\ \ \
\psi(z)_{7,4}^{2}=
-24z^3+6z,
\]
\[
\psi(z)_{7,4}^{3}=
-24z^3+7z,
\ \ \
\psi(z)_{7,4}^{4}=
-27z^3+6z,
\]
\[
\psi(z)_{7,4}^{5}=
-24z^3+5,
\ \ \
\psi(z)_{7,4}^{6}=
-20z^3+4z.
\]
\[
\psi(z)_{7,4}^{7}=
-24z^3+10z+2, 
\ \ \
\psi(z)_{7,4}^{8}=
-32z^3+10z+2,
\]
\[
\psi(z)_{7,4}^{9}=
-36z^3+10z+2,
\ \ \
\psi(z)_{7,4}^{10}=
-30z^3+10z+2.
\]
We see that there are no polynomials with the same set of zeros among them.

There are 3 graphs with $g=7$ and $g-p+r= 5$. They are shown in Fig. 9. 

\begin{figure}[h]
 \begin{center}
 \includegraphics[scale=0.5] {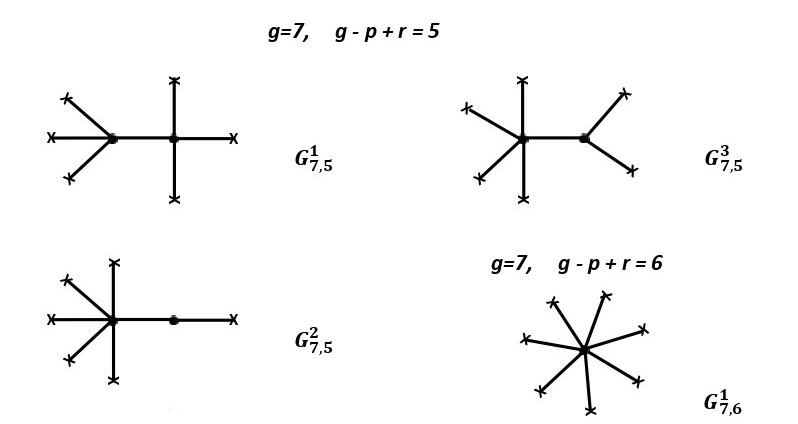}
 \end{center}
\caption{Simple connected graphs with $g=7$ and $g-p+r=5$ and $g-p+r=6$}
\end{figure}

%\end{proof}
%\end{document}

The corresponding polynomials are
\[
\psi_{7,5}^1=16z^2-1,
\ \ \
\psi_{7,5}^2=12z^2-1,
\ \ \
\
psi_{7,5}^3=15z^2-1.
\]

\[
\phi_{7,6}^1= -7z
\]

We see that the sets of zeros are different.
\end{proof}

\noindent\textbf{Acknowledgements}\\
\textit{The authors are grateful to the Ministry of Education and Science of Ukraine for the support in completing the work 'Inverse problems of finding the shape of a graph by spectral data' 
State registration number 0124U000818.}

\textit{The present research was supported by the Academy of Finland (project no. 358155). The third author is grateful to the University of Vaasa for hospitality.}

{\it  The  first and the third authors express their gratitude to NSF US for support IMPRESS-U: Spectral and geometric methods for damped wave equations with applications to fiber lasers.}
%\end{document}

\end{document}